\documentclass[11pt]{article}
\usepackage[margin = 1in]{geometry}
\usepackage[utf8]{inputenc}
\usepackage{url}
\usepackage{amsmath,amssymb,amsfonts,amsthm}
\usepackage{graphicx}

\usepackage{xcolor}
\usepackage{enumitem}
\usepackage{hyperref}
\usepackage{algorithm}
\usepackage{algpseudocode}

\RequirePackage[T1]{fontenc} \RequirePackage[tt=false, type1=true]{libertine} 
\RequirePackage[varqu]{zi4} \RequirePackage[libertine]{newtxmath}

\renewcommand{\epsilon}{\varepsilon}
\newcommand{\eps}{\varepsilon}

\newtheorem{lemma}{Lemma}
\newtheorem{theorem}{Theorem}
\newtheorem{claim}{Claim}

\newtheorem{definition}{Definition}

\newtheorem{proposition}{Proposition}

\renewcommand{\exp}[1]{{\textup{exp}}\left( #1 \right)}

\newcommand{\mc}[1]{\mathcal{#1}}

\newcommand\given[1][]{\:#1\vert\:}

\newcommand{\paren}[1]{\left ( #1 \right )}
\newcommand{\mbf}[1]{\mathbf{#1}}

\newcommand{\prob}[1]{ \mathbb{P} \left( #1 \right)}
\newcommand{\probw}[1]{ \mathbb{P}_w \left( #1 \right)}
\newcommand{\probo}[1]{ \mathbb{P}_o \left( #1 \right)}
\newcommand{\probb}[2]{ \mathbb{P}_{#1} \left( #2 \right)}

\usepackage{array}
\newcommand*{\vertbar}{\rule[-1ex]{0.5pt}{2.5ex}}

\title{Towards Better Bounds for Finding Quasi-Identifiers  \footnote{Authors are in alphabetical order. }}

\author{Ryan Hildebrant  \footnote{University of California, Irvine. Email: rhildebr@uci.edu. Part of this work was done while the author was at San Diego State University.}  \and Quoc-Tung Le \footnote{Univ Lyon, ENS de Lyon, UCBL,CNRS, Inria, LIP, F-69342, LYON Cedex 07, France. Email: quoc-tung.le@ens-lyon.fr} \and Duy-Hoang Ta \footnote{National University of Singapore. Email: hoangduybk56@gmail.com}  \and Hoa T. Vu \footnote{San Diego State University. Email: hvu2@sdsu.edu. }}
\date{}

\begin{document}

\maketitle

\begin{abstract}
We revisit the problem of finding small $\epsilon$-separation keys introduced by Motwani and Xu (2008). In this problem, the input is a data set consisting of $m$-dimensional tuples $\{x_1,x_2,\ldots,x_n\}$. The goal is to find a small subset  of coordinates that separates at least $(1-\epsilon){n \choose 2}$ pairs of tuples. When $n$ is large, they provided a fast algorithm that runs on $\Theta(m/\epsilon)$ tuples sampled uniformly at random.  We show that the sample size can be improved to $\Theta \left( {m}/{\sqrt{\epsilon}} \right)$. Our algorithm also enjoys a faster running time. 

To obtain this result, we consider a decision problem that takes a subset of coordinates $A \subseteq [m]$. It rejects if $A$ separates fewer than $(1-\eps){n \choose 2}$ pairs of tuples, and accepts if $A$ separates all ${n \choose 2}$ pairs of tuples. The algorithm must be correct with probability at least $1-\delta$ for all $2^m$ choices of $A$. We show that for algorithms based on uniform sampling:

\begin{itemize}
\item  $\Theta \left(\frac{m}{\sqrt{\eps}} \right)$ samples are sufficient and necessary so that  $\delta \leq e^{-m}$.  

\item   $\Omega \left( \sqrt{\frac{\log m}{\epsilon}} \right)$ samples are necessary so that  $\delta$ is a constant.  Closing the gap between the upper and lower bounds in this case is still an open question.
\end{itemize}

The analysis is based on a constrained version of the balls-into-bins problem whose worst case can be determined using Karush Kuhn Tucker (KKT)  conditions. We believe our analysis may be of independent interest.

We also study a related problem that asks for the following sketching algorithm: with given parameters $\alpha,k$ and $\eps$, the algorithm takes a subset of coordinates $A$ of size at most $k$ and returns an estimate of the number of unseparated pairs in $A$ up to a $(1\pm\epsilon)$ factor if it is at least $\alpha {n \choose 2}$. We show that even for constant $\alpha$ and success probability,  such a sketching algorithm must use $\Omega(mk \log \eps^{-1})$  bits of space; on the other hand, uniform sampling yields a sketch of size  $\Theta \left( \frac{mk}{\alpha \epsilon^2} \right)$  for this purpose.
\end{abstract}

\newpage

\section{Introduction} \label{sec:intro}
\paragraph{Motivation.} In large data sets, it is often important to find a small subset of attributes that identifies most of the tuples. Consider $n$ tuples $x_1,x_2,\ldots,x_n \in U^m$ each of which has $m$ coordinates (attributes) in some universe $U$. We say a subset of coordinates $A \subseteq [m]$ is a key if every pair of tuples differ by at least one coordinate value in $A$ (i.e., $A$ uniquely identifies all tuples). Motwani and Xu \cite{MX07} considered the problem of finding the minimum key $I^\star$ of a given data set. 

Let ${R \choose t}$ denote the collection of all subsets of $R$ of size $t$ and $X=\{ x_1,\ldots,x_n\}$ denote the data set, Motwani and Xu \cite{MX07} reduced the minimum key problem to the set cover problem where the ground set is ${X \choose 2}$ and each coordinate $i$ corresponds to a subset of ${X \choose 2}$ that consists of pairs of tuples differing at their $i$th coordinates. Finding the minimum key of $X$ is equivalent to finding the minimum set cover of ${X \choose 2}$. 
The set cover problem admits a $(\ln N +1)$ greedy approximation {with time complexity $O(NM^2)$ (where $N$ is the cardinality of the ground set and $M$ is the number of subsets) \cite{Young08}}. {The combination of these two ideas yields an} approach that runs in $\Theta(m^2 n^2)$ time \footnote{With a careful implementation, described in Appendix \ref{sec:min-approx-key-improvements}, one can improve the running time to $O(m^2 n \log n)$ which still depends on $n$. }. For massive data sets (i.e., large $n$), this approach is, however, costly.

\paragraph{Approximate minimum $\epsilon$-separation key via minimum set cover.}  To this end, Motwani and Xu considered the relaxed problem of finding small {\em $\eps$-separation keys}. We say that a subset of coordinates $A$ {\bf separates} $x_i$ and $x_j$ if they are different in at least one coordinate in $A$. In this problem, with high probability, we want to find a small subset of coordinates $I$ that separates at least $(1-\eps)$ fraction of pairs of tuples such that $| I | \leq \gamma |I^\star|$. In this case, one often refers to $I$ as a {\bf quasi-identifier} with separation ratio $1-\eps$. The parameter $\gamma$ controls how large $I$ is allowed to be compared to the minimum key $I^\star$. 

Their idea is to uniformly sample $\Theta \paren{m/\eps}$ pairs of tuples $R = \{ (x_{i_1}, x_{j_1}),(x_{i_2}, x_{j_2}),\ldots\}$ and solve the  set cover problem in which the ground set is $R$ and each coordinate {$i$ corresponds to a subset of $R$ that consists of pairs of tuples differing at their $i$th coordinates.} 

If $A$ separates fewer than $(1-\eps){n \choose 2}$ pairs, the probability that $A$ separates all pairs in $R$ (or equivalently, $A$ is the set cover of the described set cover instance) is at most
\[
(1-\eps)^{|R|} \leq e^{-\eps |R| } \leq e^{-10m}
\]
by choosing $|R| = {10} m/\eps$. 
By appealing to a union bound over all $2^m$ subsets of coordinates,  we guarantee that no such subset of attributes separates all pairs in $R$ with probability at least $1-e^{-5m}$  \footnote{Note that the failure probability can be set to $e^{-Km}$ for an arbitrary constant $K$.  This constant is absorbed in the big $O$ notation of the sample complexity.}. Therefore, finding a $\gamma$ approximation to the aforementioned set cover instance yields an $\epsilon$-separation key whose size is at most $\gamma | I^\star |$. 

One can attain $\gamma=1$ using the  brute-force algorithm whose running time is $2^{O\left({m}\right)}$ (but the size of the ground set is much smaller - $O(m/\epsilon)$ instead of $O(n)$). On the other hand, one can also attain $\gamma=O\left( \ln \frac{m}{\eps} \right)$ using the greedy set cover algorithm whose running time is $O(m^3/\eps )$, assuming we can compare two coordinates in constant time. This running time is more manageable as it does not depend on the size of the data set $n$. Note that sampling pairs of tuples can easily be implemented in the streaming model and the space would be proportional to the number of samples.

In this work, we mainly focus on sample and space bounds. However, we will also address running time and query time whenever appropriate. For this purpose, we make a mild assumption that one can define a total ordering of values in $U$. This is the case in most natural applications (i.e., numbers, strings, etc.); we also assume that comparing two values in $U$ takes constant time.

\paragraph{A small tweak to Motwani-Xu's algorithm.} Our goal is to improve the sample complexity with a small tweak. In particular, we sample  $\Theta\left(\frac{m}{\sqrt{\epsilon}} \right)$ tuples  uniformly at random, and let $R$ be the set of sampled tuples. We then solve the  set cover instance in which the ground set is  ${R \choose 2}$ and each coordinate corresponds to the pairs in ${R \choose 2}$ that it separates. 

In other words, the key difference between the two algorithms is that ours samples $R = \Theta(m / \sqrt{\epsilon})$ tuples and use $R \choose 2$ as ground set (for the set cover instance) while the approach of \cite{MX07} samples $R' = \Theta(m/ \epsilon)$ {pairs of tuples} of the data set and uses $R'$ as the ground set. 

We show that our approach achieves the same guarantees. In particular, for all $A \subseteq [m]$, if $A$ separates fewer than $(1-\eps){n \choose 2}$ pairs, the probability that $A$ separates all pairs in ${R \choose 2}$ is at most $e^{-m}$. While this seems like a small tweak, the analysis is significantly more involved.

If we use the greedy set cover algorithm to obtain the approximation factor $\gamma=O\left( \ln \frac{m}{\eps} \right)$, this tweak also leads to a better time complexity $O \left( \frac{m^3}{\sqrt{\eps}}  \right)$ in comparison to that of Motwani and Xu whose running time is $O(\frac{m^3}{\eps})$ . 

\paragraph{$\epsilon$-separation key filter.} 
We first consider a decision problem that captures the essence of our analysis and then describe how this yields the aforementioned improvements to finding a small $\epsilon$- separation key. We say $A$ is {\bf bad} if it separates fewer than $(1-\epsilon){n \choose 2}$ pairs of tuples. 
This problem asks for an algorithm that takes  $A \subseteq [m]$ and rejects if $A$ is bad. Furthermore, the algorithm must accept if $A$ separates all ${n \choose 2}$ pairs (i.e., $A$ is a key). If $A$ is neither bad nor a key, the algorithm can either accept or reject. The success probability is required to be at least $1-\delta$ for all $2^m$ choices of $A$. More formally,

 \[
 	\prob{\forall A \subseteq [m], \text{the algorithm is correct on $A$}} \geq 1-\delta.
 \]
 
Hereinafter, we only consider the above ``for all'' notion of success  (as opposed to the ``for each'' notion). The {\bf query time} is the time to compute the answer for a given subset of coordinates $A$. 

From the discussion above, it is fairly easy to see that the approach of Motwani and Xu solves this problem using a sample size $\Theta(m/\eps)$ with query time $O \left( \frac{m|A|}{\eps} \right)$. Specifically, we reject $A$ if it fails to separate any of $\Theta(m/\eps)$ pairs in $R'$ and accept otherwise. Our goal is to improve the sample size and query time.
 
\paragraph{Non-separation estimation.}  We also consider a related approximate counting problem. Let $\Gamma_A$ be the number of pairs of tuples that are not separated by $A$. Let $\alpha,\eps$ and $k$ be some given parameters, the {\em non-separation estimation} problem asks for an algorithm that takes a subset of coordinates $A \subset [m]$ as the input where $|A| \leq k$. If  $\Gamma_A \geq \alpha {n \choose 2}$, the algorithm must return an estimate $\hat{\Gamma}_A$ of  $\Gamma_A$ such that
\[
	\hat{\Gamma}_A \in [(1-\epsilon) \Gamma_A, (1+\epsilon)\Gamma_A]~.
\]
Otherwise, the algorithm may output ``small''. The algorithm must return correct answers for all queries $A \in [m]$ (where $|A| \leq k$) with probability at least $1-\delta$.

\paragraph{Sketching algorithms and algorithms based on uniform sampling.} A sketch $S$ of a data set $X$, with respect to some problem $F$, is a compression of $X$ such that given only access to $S$, one can solve the problem $F$ on $X$. Often, the primary objective is to minimize the size of $S$.

Given $X$, an algorithm based on uniform sampling is a special case of sketching algorithms in which $S$ is a set of items drawn uniformly at random from $X$. 

\paragraph{Main results and techniques.} For the $\epsilon$-separation key filter problem, our main results are as follows. We first propose a better sketching algorithm based on uniform sampling in terms of sample size and query time. We then establish its sampling complexity lower bounds. In certain regime, we show that our strategy is optimal in term of sampling complexity.

\begin{theorem}[Main result 1]\label{thm:main1}
Consider the $\epsilon$-separation key filter problem. For algorithms based on uniform sampling, if $n \geq Km/\eps$ for some sufficiently large constant $K$, then: 
\begin{itemize}
\item  $\Theta \paren{\frac{m}{\sqrt{\eps}}}$ samples are sufficient and necessary so that  $\delta \leq e^{-m}$; furthermore, the query time is $O\paren{\frac{m|A|}{\sqrt{\eps}} \log \frac{m}{{\eps}}}$ where $A$ is the subset of coordinates being queried.
\item   $\Omega \left( \sqrt{\frac{\log m}{\epsilon}} \right)$ samples are necessary so that  $\delta$ is a constant.  \end{itemize}
\end{theorem}

The analysis of the upper bound is heavily non-trivial. At a high level, it uses Karush–Kuhn–Tucker (KKT) \cite[Chapter 12]{NoceWrig06} conditions to identify the worst-case input. With some further arguments, we can apply the birthday problem \cite[Page 45]{motwani1995randomized} to show that in this worst-case input, $\Theta(m/\sqrt{\eps})$ sample size suffices.

This result leads to the following improvements for the approximate minimum $\epsilon$-separation key problem in both sample size and running time.

\begin{proposition} \label{prop:approximate-min-key}
There exists an algorithm based on uniform sampling that solves the approximate minimum $\eps$-separation key problem with a sample size $\Theta \left( \frac{m}{\sqrt{\eps}} \right)$. Furthermore, if the approximation factor $\gamma = O\left(\ln \frac{m}{\eps} \right)$, the running time is $O \left( \frac{m^3}{\sqrt{\eps}}\right)$.
\end{proposition}

For the non-separation estimation problem, we show that a sketching algorithm based on uniform sampling is optimal in terms of $k$ and $m$ up to a logarithmic factors.

\begin{theorem}[Main result 2]\label{thm:main2}
Consider the non-separation estimation problem, the following holds.
\begin{itemize}
\item There exists an algorithm based on uniform sampling that requires a sample size $\Theta \paren{\frac{k \log m}{\alpha \eps^2}}$.
\item Suppose $m \geq k/\sqrt{\eps}$ and $\alpha$ is a constant. Such a sketching algorithm must use $\Omega(mk \log (1/\eps))$ space.
\end{itemize}
\end{theorem}

Note that a sample size $\Theta \paren{\frac{k \log m}{\alpha \eps^2}}$ requires $\Omega \paren{\frac{km \log m \log |U|}{\alpha \eps^2 }}$ bits of space. Hence, for constant $\alpha$, the upper and lower bounds are tight in terms of $k$ and $m$ up to a logarithmic factor. The upper bound in the theorem above is a direct application of Chernoff and union bounds. The lower bound is based on an encoding argument. Liberty et al. \cite{LibertyMTU16} also used an encoding argument to prove a space lower bound for finding frequent itemset although the technical details are different. 

\paragraph{Further applications.} Motwani and Xu \cite{MX07} have highlighted multiple applications of this problem, which we will summarize here. One example is using this problem as a fundamental tool to identify various dependencies or keys in noisy data. This problem can also aid in comprehending the risks associated with releasing data, specifically the potential for attacks that may re-identify individuals.

Small quasi-identifiers are crucial information to consider from a privacy perspective because they can be utilized by adversaries to conduct linking attacks. The collection of attribute values may come with a cost for adversaries, leading them to seek a small set of attributes that form a key.

This problem also has applications in data cleaning, such as identifying and removing fuzzy duplicates resulting from spelling mistakes or inconsistent conventions \cite{AnanthakrishnaCG02, ChaudhuriGGM03}. Moreover, quasi-identifiers are a specific case of approximate functional dependency \cite{KivinenM92, PfahringerK95}. The discovery of such quasi-identifiers can be valuable in query optimization and indexing \cite{giannella2002using}.

\paragraph{Related work.} There has been recent work on sketching and streaming algorithms as well as lower bounds for multidimensional data. In this line of work,  we want to compute a sketch of the input data set such that given only access to the sketch, we can approximate statistics of the data restricted to some query subset(s) of coordinates $A$; in this setting, $A$ is provided after the sketch has been computed. Some example problems include heavy-hitters and frequent itemset \cite{KvetonMVX18,LibertyMTU16}, frequency estimation \cite{CormodeDW21}, and box-queries \cite{FriedmanS22}. 

\paragraph{Preliminaries.} We provide some useful tools and facts as well as common notations for the rest of the paper here.

\emph{Chernoff bound.} We often rely on the following version of Chernoff bound:
\begin{theorem}[Chernoff bound \cite{wiki:Chernoff_bound}]
\label{theorem:chernoffbound}
Let $X_1,X_2,\ldots,X_N$ be independent Bernoulli random variables such that $\prob{X_i=1} = p$. Let $X = \sum_{i=1}^N X_i$ and $\mu = \mathbb{E}[X] = pn$. Then, for all $\epsilon > 0$,
\[
\prob{|X - p N| \geq \epsilon pN } \leq 2 e^{- \frac{\epsilon^2}{2 +\epsilon} \mu} \implies
\prob{X \leq \frac{pN}{2}} \leq 2 e^{-0.1 pn}~.
\]
When $\epsilon \geq 2$, we have:
\[
\prob{|X - p N| \geq \epsilon pN } \leq 2 e^{- \frac{\epsilon}{2 } \mu}~.
\]
\end{theorem}

\emph{The birthday problem.} If one throws $q$ balls (people) into $N$ bins (birthdays) uniformly at random  \footnote{This assumption is indeed unnecessary. If the distribution is non-uniform, one can show that the probability of collision only increases  \cite{mcconnell2001inequality,steele2004cauchy}.}, one can show that the probability that there is a {\bf collision} (i.e., there is a bin that contains at least two balls) is approximately $1-e^{-\Theta(q^2/N)}$. The following can be found in various textbooks and lecture notes (for example, see \cite[Page 45]{motwani1995randomized}).

\begin{theorem}[The birthday problem] \label{thm:birthday} Let $C(N, q)$ be the probability that there is at least one collision if one throws $q$ balls into $N$ bins uniformly at random. We have:
\[
C(N,q) \geq 1-e^{\frac{-q(q-1)}{2N}}~.
\]
This implies that for the non-collision probability to be less than $\delta^{\star}$, we can set 
\[
q \geq \frac{1}{2} \left(1+\sqrt{8N \log\frac{1}{\delta^\star} +1} \right) \geq 4 \sqrt{N \log \frac{1}{\delta^\star}}~.
\]
\end{theorem}

\emph{Notation and assumptions.} For convenience, hereinafter, we let $K$ be a universal sufficiently large constant.   Unless stated otherwise, we may round $\epsilon$ down to the nearest power of $1/4$, i.e., $\epsilon = 1/4^t$ where $t$ is an integer so that $1/\epsilon$ and $1/\sqrt{\epsilon}$ are integers. This does not change the complexity as the extra constant is absorbed in the big-$O$ notation.  Unless stated otherwise, $\log n$ refers to $\log_2 n$.

\paragraph{Organization.} Section \ref{sec:improved-upperbound} provides  the improved upper bound and lower bounds for the $\eps$-key filter problem in Theorem \ref{thm:main1}. Section \ref{sec:non-separation-est} provides the upper and lower bounds for the non-separation estimation problem in Theorem \ref{thm:main2}. We also implemented a small experiment showing the efficiency of the new approach in Section \ref{sec:implementation}. 

The resulted improvements to the approximate minimum $\eps$-separation key in Proposition \ref{prop:approximate-min-key} are explained in Appendix \ref{sec:min-approx-key-improvements}.

All omitted proofs can be found in Appendix \ref{sec:omitted-proofs}. 

\section{Improved Upper and Lower Bounds for  $\epsilon$-Separation Key Filter}\label{sec:improved-upperbound}
\subsection{An Improved Upper Bound}

\paragraph{Algorithms.} We will prove the first part of Theorem \ref{thm:main1} in this subsection. The detailed algorithm is {given in Algorithm \ref{algorithm0}}. Our algorithm samples $R = \Theta(m/\sqrt{\eps})$ tuples without replacement uniformly at random. For any given $A \subseteq [m]$, if $A$ fails to separate any pair in ${R \choose 2}$, then reject $A$; otherwise, accept. We will show that this algorithm is correct with probability at least $1-e^{-m}$.

\begin{algorithm}[H]
	\centering
	\caption{An improved sampling-based algorithm for $\epsilon$-separation key filter} 
	\label{algorithm0}
	\begin{algorithmic}[1]
		\Require Data set $X = \{x_1, \ldots, x_n\}, x_i \in U^m$, $\epsilon$. 
		\State Sample without replacement $R = \Theta\paren{m/\sqrt{\epsilon}}$ tuples.
		\State Given $A \subseteq [m]$, accept $A$ if and only if {$A$ separates all $R \choose 2$ pairs of samples}.					
	\end{algorithmic}
\end{algorithm}

\paragraph{Improvement to minimum $\epsilon$-separation key.} Note that if this algorithm for $\eps$-separation keys filter is correct, the improvement for the minimum $\eps$-separation key problem in terms of sample size is immediate. The improved running time is less trivial and requires some careful bookkeepings; we defer the discussion to Appendix \ref{sec:min-approx-key-improvements}.

\paragraph{Analysis.} Let us visualize this problem in terms of  auxiliary graphs. Consider a subset of attributes $A$. If $A$ fails to separate $x_i$ and $x_j$ then we draw an edge between $x_i$ and $x_j$. Note that if $A$ fails to separate $x_i$ and $x_j$, and $A$ fails to separate $x_j$ and $x_l$, then $A$ also fails to separate $x_i$ and $x_l$. Therefore, the graph $G_A$ consists of disjoint cliques. The goal is to sample an edge in $G_A$  (or equivalently, sample two tuples not separated by $A$)  if $A$ fails to separate $\epsilon$ fraction of the pairs (in this case, we call $A$ {\em bad}).

We want to argue that if $A$ is bad, then we sample two distinct vertices in the same clique in $G_A$ with probability at least $1-e^{-10m}$. By an application of the union bound over $2^m$ subsets of attributes, we are able to discard all bad subsets of attributes with probability at least $1-e^{-m}$.

We note that if $A$ is bad, $G_A$ has at least $\epsilon {n \choose 2}$ edges. Thus, the problem can be rephrased as follows. Given a graph of $n$ vertices and at least $\epsilon {n \choose 2}$ edges that consists of disjoint cliques, what is the smallest number of vertices that one needs to sample uniformly at random such that the induced graph has at least one edge with probability at least $1-e^{-10m}$?

It remains to show that sampling $r = \Theta\left(\frac{m}{ \sqrt{\epsilon}} \right)$ tuples (or vertices) is sufficient. The key observation is that this problem is fairly similar to the { birthday problem}. Of course, this is not quite the same since we want to sample two {\em distinct} vertices from cliques of size at least two. We can simply sample without replacement to avoid this issue. However, for a cleaner analysis, our analysis is based on sampling with replacement and relate the two.

There could be at most $n$ cliques so we use a vector $s=(s_1,s_2,\ldots,s_n)$ to denote the cliques' sizes where zero-entries are empty cliques. Recall that $ \sum_{i=1}^n \frac{s_i(s_i-1)}{2} \geq \epsilon {n \choose 2}$ which implies $\sum_{i=1}^n s_i^2 \geq \epsilon n^2/4 $ for sufficiently large $n$. 

We relax the problem so that $s_i$ can be a non-negative real number (i.e., $s_i \in \mathbb{R}_{\geq 0}$). Consider following set of points $s \in \mathbb{R}_{\geq 0}^n$ such that
\begin{align}
    & \sum_{i=1}^n s_i^2 \geq \epsilon n^2/4 \label{eq:edge-constraint} \\
    & \sum_{i=1}^n s_i = n \label{eq:vertex-constraint} \\
    & s_i \in \mathbb{R}_{\geq 0} \text{, for all $i=1,2,\ldots,n$}. \label{eq:nonnegative-constraint} 
\end{align}

We use $\mathcal{P}$ to denote the set of all $s$ that satisfies all three constraints. For convenience, we can think of cliques as colors and each vertex that we sample uniformly at random is a ball whose color is distributed  according to the multinomial distribution $\mathcal{D}_s = (\frac{s_1}{n}, \frac{s_2}{n},\ldots,\frac{s_n}{n})$ where $s \in \mathcal{P}$.

We are interested in how large $r$ should be such that the probability of having two balls with the same color is at least $1-e^{-10m}$. In particular, let $\probb{r,\mathcal{D}_s}{\xi}$  be the probability that at no two balls have the same color (we refer to this event $\xi$ as {\em non-collision}) if we draw $r$ balls whose colors follow the distribution $\mathcal{D}_s$. 
 
Let us fix $r$ and figure out $s$ that maximizes the non-collision probability. Without constraint \eqref{eq:edge-constraint}, it is easy to see that the non-collision probability is largest when the color distribution is uniform \cite{mcconnell2001inequality,steele2004cauchy}. With constraint \eqref{eq:edge-constraint}, one might suspect that this is still true. That is the non-collision probability is largest when all \emph{non-zero} $s_i$ have the same value (i.e., all $1/\epsilon'$ non-zero entries have the same value $\epsilon' n$ where $\epsilon' = \epsilon/4$). This intuition is however {\em incorrect}; one can see an example in Appendix \ref{sec:Appendix-examples}. In this case, the problem indeed becomes more complicated.

Given the distribution $\mathcal{D}_s$, the probability that we do not have a collision after sampling $r$ vertices uniformly at random is

\[
\probb{r,\mathcal{D}_s}{\xi} = \frac{r!}{n^r} \sum_{1 \leq j_1 <  j_2 < \ldots < j_r \leq n} s_{j_1} s_{j_2} \ldots s_{j_r} := \frac{r!}{n^r} f_r(s)~.
\]

Suppose we sample without replacement, the probability $\probb{r,\mathcal{D}_s,\diamond}{\xi}$ that we do not have a collision is
\[
\probb{r,\mathcal{D}_s,\diamond}{\xi}  := \frac{r!}{n(n-1)(n-2)\ldots(n-r+1)} f_r(s)~.
\]
As a consequence, we have:
\begin{align*} 
	\probb{r,\mathcal{D}_s,\diamond}{\xi} = \frac{n^r}{n(n-1)\ldots(n - r + 1)} \probb{r,\mathcal{D}_s}{\xi} 
\end{align*}
The next claim relates $\probb{r,\mathcal{D}_s,\diamond}{\xi} $ and $\probb{r,\mathcal{D}_s}{\xi} $. 

\begin{claim}\label{lem:non-replacement}
For $n  > \frac{r(r-1)}{m}+r-1 = \Theta(\frac{r^2}{m} +r )$, we have 
\[\probb{r,\mathcal{D}_s,\diamond}{\xi}  < e^{m} \probb{r,\mathcal{D}_s}{\xi} ~.\]
\end{claim}
\begin{proof}
Note that $n  > \frac{r(r-1)}{m}+r-1 $ implies $ \frac{r(r-1)}{n-r+1} < m$. Thus,
\begin{equation}
\label{eq:boundedterm}
\begin{aligned}
\frac{n^r}{n(n-1)(n-2)\ldots(n-r+1)} & \leq \left(\frac{n}{n-r+1} \right)^r= \left(1 + \frac{r - 1}{n - r + 1} \right)^r   \leq e^{\frac{r(r-1)}{n-r+1}} \leq e^{m}~. \qedhere
\end{aligned}
\end{equation}

\end{proof}

We will later set $r = \Theta \left( \frac{m}{\sqrt{\epsilon}} \right)$; the condition $n  > \frac{r(r-1)}{m}+r-1$ implies that $n  \geq \frac{Km}{\epsilon}$. We are now interested in maximizing $f(s)$ for fixed $r$. Let
\begin{align*}
s^{\star} = \max_{s \in \mathcal{P}} \probb{r,\mathcal{D}_s}{\xi}  = \max_{s \in \mathcal{P}} f_r(s)~.
\end{align*}

We make the following key observation.

\begin{lemma} \label{lem:key-lemma}
If $4 \leq r \leq 1+(1-\sqrt{\epsilon}/2)n$, then the optimal $s^\star = \max_{s \in P} f_r(s)$ satisfies the following: all non-zero entries in $s^\star$ have at most two distinct values.
\end{lemma}
\begin{proof}
Since $r$ is being fixed, we drop $r$ from $f_r$ for a cleaner presentation. Let 
\begin{equation}
	\label{eq:extremecase}
	\tilde{s}=(\sqrt{\epsilon} n/2,\underbrace{1, 1,\ldots, 1}_{(1-\sqrt{\epsilon}/2)n \text{ times }},0,0,\ldots,0)~.
\end{equation}
In fact, $\tilde{s}$ is feasible as it satisfies all constraints \eqref{eq:edge-constraint} - \eqref{eq:nonnegative-constraint}. We note that any $s$ that has fewer than $r$ non-zero entries cannot be the optimal value since $f(s) = 0$ while $f(\tilde{s}) > 0$ since it has $(1-\sqrt{\epsilon}/2)n+1$ non-zero entries and therefore at least one term in $f(\tilde{s})$ must be positive.

Our proof uses the KKT conditions with LICQ regularity condition as its core. We invite readers who are not familiar with this technical theorem to visit Appendix \ref{sec:kkt} for full details.

Let $s$ be any local optimal. First, we assume that LICQ holds at $s$. Based on the KKT conditions (Theorem \ref{theorem:KKT}), there exist some constants $\mu \geq 0$,  $\eta \in \mathbb{R}$, and $\{ \lambda_i\}_{i \in [n]} \geq 0$, such that for any local maxima, we have 
\begin{align} \label{eq:stationarity}
\nabla f(s) = \mu \nabla \left(\sum_{i=1}^n s_i^2 \right) +    \sum_{i=1}^n \lambda_i \nabla  (s_i) + \eta  \nabla \left(\sum_{i=1}^n s_i \right)
\end{align}
and 
\begin{align}
\label{eq:slackness}
s_i > 0  \implies \lambda_i=0~.
\end{align} 

Equations \eqref{eq:stationarity} and \eqref{eq:slackness} correspond to the stationarity and complementary slackness conditions (see Appendix \ref{sec:kkt}). Each coordinate of the right hand side of Eq. \eqref{eq:stationarity} has the following form 
\begin{align}
2 \mu s_i + \lambda_i + \eta~.
 \end{align}

Each coordinate $i$ of $\nabla f(s)$ is as follows:
\begin{align}
\frac{\partial f(s)}{\partial s_i} & =  \sum_{\substack{1 \leq j_1 < j_2 < \ldots j_{r-1} \leq n \\ j_1,j_2,\ldots,j_{r-1} \neq i}}   s_{j_1} s_{j_2} \ldots s_{j_{r-1}}~.
\end{align}

 Therefore, for each $i=1,2,\ldots,n$:
\begin{align}\label{eq:lagrange1}
\sum_{\substack{1 \leq j_1 < j_2 < \ldots < j_{r-1} \leq n \\ j_1,j_2,\ldots,j_{r-1} \neq i}}  s_{j_1} s_{j_2} \ldots s_{j_{r-1}} - (2 \mu s_i + \lambda_i + \eta)= 0~ .
\end{align}

Consider any $s_a > 0$ and $s_b > 0$. Note that $\lambda_a =\lambda_b = 0$. We have:
\begin{equation}
\begin{aligned}
 \bigg( \sum_{\substack{1 \leq j_1 < j_2 < \ldots < j_{r-1} \leq n \\ j_1,j_2,\ldots,j_{r-1} \neq a}}  s_{j_1} s_{j_2} \ldots s_{j_{r-1}} \bigg) - 2\mu s_a 
& = \bigg(  \sum_{\substack{1 \leq j_1 < j_2 < \ldots < j_{r-1} \leq n \\ j_1,j_2,\ldots,j_{r-1} \neq b}}  s_{j_1} s_{j_2} \ldots s_{j_{r-1}} \bigg) - 2\mu s_b  \\
\iff s_b \bigg( \sum_{\substack{1 \leq j_1 < j_2 < \ldots < j_{r-2} \leq n \\ j_1,j_2,\ldots,j_{r-2} \neq a, b}}  s_{j_1} s_{j_2} \ldots s_{j_{r-2}} \bigg) - 2\mu s_a 
& = s_a \bigg( \sum_{\substack{1 \leq j_1 < j_2 < \ldots < j_{r-2} \leq n \\ j_1,j_2,\ldots,j_{r-2} \neq a,b}}  s_{j_1} s_{j_2} \ldots s_{j_{r-2}} \bigg) - 2\mu s_b  \\
\iff (s_b-s_a) \sum_{\substack{1 \leq j_1 < j_2 < \ldots < j_{r-2} \leq n \\ j_1,j_2,\ldots,j_{r-2} \neq a,b}} s_{j_1} s_{j_2} \ldots s_{j_{r-2}} & = -2\mu(s_b - s_a)~.
\end{aligned}
\end{equation}

Thus, either $s_a = s_b$ or 
\begin{align}\label{eq:lagrange3}
 \sum_{\substack{1 \leq j_1 < j_2 < \ldots < j_{r-2} \leq n \\ j_1,j_2,\ldots,j_{r-2} \neq a,b}} s_{j_1} s_{j_2} \ldots s_{j_{r-2}} & = -2\mu_1~.
\end{align}
If $s_a \neq s_b$, then consider any other entry $s_c > 0$. It must be the case that either $s_c \neq s_a$ or $s_c \neq s_b$. Without loss of generality, suppose $s_c \neq s_b$. With the same derivation, we have:
\begin{align}\label{eq:lagrange4}
 \sum_{\substack{1 \leq j_1 < j_2 < \ldots < j_{r-2} \leq n \\ j_1,j_2,\ldots,j_{r-2} \neq b,c}} s_{j_1} s_{j_2} \ldots s_{j_{r-2}} & = -2\mu_1~.
 \end{align}
 From Eq. \eqref{eq:lagrange3} and Eq. \eqref{eq:lagrange4}, we have:
 \begin{equation}
\begin{aligned}\label{eq:lagrange5}
 \sum_{\substack{1 \leq j_1 < \ldots < j_{r-2} \leq n \\ j_1,\ldots,j_{r-2}  \neq b,c}} s_{j_1} \ldots s_{j_{r-2}} & =  \sum_{\substack{1 \leq j_1 < \ldots < j_{r-2} \leq n \\ j_1,\ldots,j_{r-2} \neq a,b}} s_{j_1}\ldots s_{j_{r-2}}  \\
s_a \sum_{\substack{1 \leq j_1 < \ldots < j_{r-3} \leq n \\ j_1,\ldots,j_{r-3} \neq a,b,c}} s_{j_1} \ldots s_{j_{r-3}} & = s_c \sum_{\substack{1 \leq j_1 < \ldots < j_{r-3} \leq n \\ j_1,\ldots,j_{r-3} \neq a,b,c }} s_{j_1} \ldots s_{j_{r-3}} \\
s_a & = s_c~.
\end{aligned} 
\end{equation}
 In the last step, we divide both sides by 
 \begin{align}
 \sum_{\substack{1 \leq j_1 < j_2 < \ldots < j_{r-3} \leq n \\ j_1,j_2,\ldots,j_{r-3} \neq a,b,c }} s_{j_1} s_{j_2} \ldots s_{j_{r-3}} > 0~.
 \end{align}
 This is because we assume that at least $r \geq 4$ entries are non-zero. Thus, all entries in $s$ can either take value 0, $a$, or $b$.

It remains to deal with the case where LICQ does not hold at $s$. We denote by $\mbf{0}$ and $\mbf{1}$ the vectors in which all entries are zeros and ones respectively; furthermore, $\mbf{1}_i$ denotes the $i$th canonical vector.  {Let $\mc{A}(s)$ be the (active) set of constraints whose equality hold (see  Appendix \ref{sec:kkt} and Definition \ref{def:activeset} for a formal definition)}. The active set $\mc{A}(s)$ can contain three following types of constraints:
\begin{enumerate}
	\item Constraint \eqref{eq:edge-constraint}: $\sum_{i = 1}^n s_i^2 \geq \epsilon n^2 / 4$. The gradient corresponding to this constraint is $2s$.
	\item Constraint \eqref{eq:vertex-constraint}: $\sum_{i = 1}^n s_i = n$. The gradient corresponding to this constraint is $\mbf{1}$. Note that constraint \eqref{eq:vertex-constraint} must be in $\mc{A}(s)$ by definition.
	\item Constraints \eqref{eq:nonnegative-constraint}: $s_i \geq 0$. The gradient corresponding to this type of constraints is $\mbf{1}_i$. 
\end{enumerate}
Denote by $A: = \{i \mid s_i = 0\}$ the set of indices whose components of $s$ are zero. If $\mc{A}(s)$ does not contain Constraint \eqref{eq:edge-constraint} and LICQ does not hold means, then there exist $\gamma \in \mathbb{R}^{|A| + 1}, \gamma \neq 0$ such that:
\begin{equation*}
	\mbf{0} = \gamma_0\mbf{1} + \sum_{i \in A} \gamma_i\mbf{1}_i
\end{equation*}
This happens if and only if $\mc{A}(s)$ contains all constraints \eqref{eq:nonnegative-constraint}, i.e, $s = \mbf{0}$, which contradicts Constraint \eqref{eq:vertex-constraint}. Therefore, $\mc{A}(s)$ must contain Constraint \eqref{eq:edge-constraint}. By our assumption {that LICQ does not hold}, there exists a vector $\gamma \in \mathbb{R}^{|A| + 1}$ such that:
\begin{equation*}
	2s = \gamma_0 \mbf{1} + \sum_{i \in A} \gamma_i \mbf{1}_i
\end{equation*}

As a consequence, $s_i = \gamma_0/2$ 
for all $ i \notin A$. Thus, there is at most one distinct value among non-zero entries of $s$ when LICQ does not hold. That concludes our proof.
\end{proof}

The next lemma shows that we can still apply the birthday problem for $r= \Theta \paren{ \frac{m}{\sqrt{\epsilon}}}$.

\begin{lemma} \label{lem:optimal-property}
Suppose the non-zero entries in $s \in \mathcal{P}$ ($\mathcal{P}$ defined by Constraints \eqref{eq:edge-constraint}-\eqref{eq:nonnegative-constraint}) have at most two distinct values. By drawing $r= \Theta \paren{ \frac{m}{\sqrt{\epsilon}}}$ balls uniformly at random from $\mathcal{D}_s$, the probability of no two balls having the same color is at most $1-e^{-20m^2} - e^{-20m}$.
\end{lemma}
\begin{proof}
Consider the case where there are two distinct values among non-zero entries in $s$. Let these two distinct values  be $a$ and $b$. We say a color $i \in [n]$ is in group $A$ if $s_i = a$; otherwise $i$ is in group $B$.  Define $k_a := |A|, k_b: = |B|$, we have:
\[
k_a a^2 + k_b b^2 \geq \frac{\epsilon n^2}{4}~.
\]
Thus either $k_a a^2 \geq \frac{\epsilon n^2}{8}$ or $k_b b^2 \geq \frac{\epsilon n^2}{8}$. Without loss of generality, suppose $k_a a^2 \geq \frac{\epsilon n^2}{8}$. Then
\begin{align*}
a & \geq \frac{\sqrt{\epsilon}n}{\sqrt{8 k_a}} \iff k_a a  \geq \frac{\sqrt{k_a \epsilon}n}{\sqrt{8}}~.
\end{align*}

Thus, in expectation, each random ball drawn uniformly at random has color is in group $A$ with probability at least $\frac{\sqrt{\epsilon k_a}}{\sqrt{8}}$. 

Let $K$ be some sufficiently large constant, if we draw at least $ K  m \sqrt{k_a }$ balls uniformly at random restricted to group-$A$ colors, by the birthday problem (i.e., Theorem \ref{thm:birthday}), the probability that no two balls share the same color is at most
 $e^{-20  m^2}$. 

Suppose we sample $\frac{2\sqrt{8} Km}{\sqrt{\epsilon}}$ balls uniformly at random. In expectation, the number of balls whose colors are in group $A$ in the sample is at least
\[
\frac{2 \sqrt{8} K  m}{\sqrt{\epsilon}} \times \frac{\sqrt{\epsilon k_a}}{\sqrt{8}} = 2 K m \sqrt{k_a} ~.
\]
By Chernoff bound (Theorem \ref{theorem:chernoffbound}), the probability that we sample fewer than $K m \sqrt{k_a} $ balls whose colors are in group $A$ is at most $2e^{-0.1 \times  2 K m \sqrt{k_a}} \leq e^{-20m}$.
Thus, we get two balls with the same color with probability at least 
\begin{align*}
1-e^{-20m^2} - e^{-20m}~.
\end{align*}

The case of one distinct value can be dealt with similarly by simply choose $b$ to be some arbitrary value and $k_b = 0$. The same argument is still valid and this concludes our proof.
\end{proof}

Let $r =  \frac{Cm}{\sqrt{\epsilon}}$ for some sufficiently large constant $C$.  From Lemmas \ref{lem:key-lemma} and \ref{lem:optimal-property}, for any distribution $\mathcal{D}_s$, 
\begin{align*}
\probb{r, \mathcal{D}_s}{\xi} & \leq \probb{r, \mathcal{D}_{s^\star}}{\xi} \leq 1-e^{-20m^2} - e^{-20m}~.
\end{align*}

\paragraph{Putting it all together.}If we sample tuples without replacement, for each bad subset of coordinates $A \subseteq [m]$, the probability that we do not sample an edge in $G_A$ is at most
\[
e^m (e^{-20m^2} + e^{-20m}) = e^{-20 m^2 + m} + e^{-19 m} < e^{-10m}~.
\]
Taking a union bound over at most $2^m$ bad subsets $A \subseteq [m]$, we have proved that the probability of failing to detect a bad subset of coordinates is at most $2^m e^{-10 m} < e^{-m}$. 

\paragraph{Query time.} Recall that each coordinate has values in $U$. If $U$ has a total ordering, then we can sort the tuples in $R$ using $O(\frac{m}{\sqrt{\eps}} \log \frac{m}{\eps})$ comparisons each of which takes $O(|A|)$ time to detect duplicate(s). Thus, the query time is $O\paren{ \frac{m|A|}{\sqrt{\eps}} \log \frac{m}{\eps} }$.

\subsection{A lower bound  $\Omega\left(\sqrt{\frac{\log m}{\epsilon}}\right)$ for constant  failure probability}

In this section, we consider the following question: What is the lower bound on the number of tuples one needs to sample (uniform sampling, with or without replacement) such that the probability of failure of Algorithm \ref{algorithm0} is smaller than a fixed constant $\delta < 1$. To this end, we construct a model to get a lower bound $\Omega\left(\sqrt{\frac{\log m}{\epsilon}}\right)$ for this question. 
\begin{lemma}
	\label{lem:lbfixeddelta}
	There exists a data set where one needs to sample with replacement (resp. without replacement) at least $\Omega\left(\sqrt{\frac{\log m}{\epsilon}}\right)$ tuples to reject all bad subsets $A$ with probability $1/e$ (resp. $2/e$).
\end{lemma}
\begin{proof}[Proof sketch]
	We provide a proof sketch for the case of sampling with replacement here. The full proof is deferred to Appendix  \ref{subse:lemlbfixeddelta}.
We consider the data set $\mc{D}:= \{1, \ldots, \lfloor 1/ \epsilon\rfloor\}^m$. In this data set, the following holds:
\begin{enumerate}
	\item All subsets of size one are \emph{bad} (i.e they separate fewer than $(1- \epsilon) {n \choose 2}$ pairs).
	\item Sampling a tuple uniformly at random from $\mc{D}$ is equivalent to sample each coordinate \emph{i.i.d} from the uniform multinomial distribution on $\{1, \ldots, \lfloor 1/ \epsilon\rfloor\}$.
\end{enumerate}

	Denote $R_{A,r}$ the event that a bad subset $A$ is detected after $r$ samples (i.e., we sample a pair of tuples that are not separated by $A$). We derive our lower bound by bounding $\prob{\bigcap_{A : |A| = 1} R_{A,r}}$, which is the probability that one can detect all the bad subsets of cardinality one. Let $q=1/\eps$. We have
\begin{align*}
	\prob{\bigcap_{A : |A| = 1} R_{A,r}} &= \prod_{i =1}^m \prob{R_{\{i\},r}}= \prob{R_{\{i\},r}}^m \leq \left(1 - \prod_{i = 0}^{r-1} (1 - i / q) \right)^m~.
\end{align*} 

For $m \leq 2^{(1/\eps)}, r = \sqrt{\frac{\log m}{\epsilon}}$, the above can be upper bounded as 
\begin{align*}
	\prob{\bigcap_{A : |A| = 1} R_{A,r}} \leq \left( 1 - \exp{-r^2 / q} \right)^m \leq \left(1 - 1/m \right)^m \leq 1/e. &&& \qedhere
\end{align*} 
		
\end{proof}

\subsection{A lower bound $\Omega \paren{\frac{m}{\sqrt{\eps}}}$ for  $e^{-m}$ failure probability}
For a constant success probability $\delta$, our analysis leaves a gap between the upper and lower bound of the sampling complexity (i.e the upper bound is $\Theta\left(\frac{m}{\sqrt{\epsilon}}\right)$ while the lower bound is $\Omega\left(\sqrt{\frac{\log m}{\epsilon}}\right)$). However, for large success probability (i.e., $1 - e^{-m}$), one can actually show that our analysis is tight. In fact, this lower bound holds even if one only needs to tell if a single coordinate is a $\eps$-separation key.

\begin{lemma}
\label{lem:lblargedelta}
There exists a data set where one needs to sample without replacement at least $\Omega\left(\frac{m}{\sqrt{\epsilon}}\right)$ to reject a bad subset $A$ \emph{with} probability $1 - e^{-m}$.
\end{lemma}
\begin{proof}[Proof sketch]
We build a data set $\mc{D}:= \{x_1, \ldots, x_n\}$ satisfying the following properties: A) If we sample a tuple uniformly at random, its first coordinate follows the multinomial distribution given in Equation \eqref{eq:extremecase} and B) The remaining coordinates can be chosen arbitrarily as long as there exists a key for $\mc{D}$. For this data set, one can show that:
\begin{enumerate}
	\item Coordinate $\{1\}$ is bad.
	\item The graph $G_{\{1\}}$ (the auxiliary graph corresponding the first coordinate) has one cluster $C$ of size $\Theta(\sqrt{\epsilon} n)$ and $n - \Theta(\sqrt{\epsilon}n)$ clusters of size one. 
\end{enumerate}

To detect that $\{1\}$ is bad with probability $1 - e^{-m}$, one needs to sample at least two vertices (or tuples) in the largest cluster $C$ with the same probability. This requires $\Theta(m/\sqrt{\eps})$ samples since the probability of sampling a tuple in $C$ is $\Theta(\sqrt{\epsilon})$. A detailed proof is can be found in Appendix \ref{subsec:lemlblargedelta}.
\end{proof}

\section{Estimating Non-Separation} \label{sec:non-separation-est}

In this section, we prove Theorem \ref{thm:main2}. Specifically, we present space upper and lower bounds for estimating non-separation. The upper bound is based on sampling $O\paren{\frac{k \log m}{\alpha \eps^2}}$ pairs of tuples uniformly at random. 

We then prove a lower bound $\Omega \left( {km} \log \frac{1}{\eps} \right)$ on the sketch's size for constant $\alpha$. Note that this implies a lower bound $\Omega{\paren{k \log \frac{1}{\eps}}}$  on the sample size since each sample requires $\Omega(m)$ bits of space. Therefore, for constant $\alpha$, the simple sketching algorithm based on uniform sampling is tight in terms of $k$ and $m$, up to poly-logarithmic factor.

\subsection{Upper bound} We give a simple upper bound based on random sampling. Let $\mathcal{D}$ denote the data set $\{x_1,x_2,\ldots,x_n\}$. The algorithm samples $\frac{K k \log m}{\alpha \epsilon^2}$ pairs of tuples $(y_1,z_1), (y_2,z_2),\ldots \in {\mathcal{D} \choose 2}$ uniformly at random for some sufficiently large constant $K$. For a fixed $A \subseteq [n]$, we use $D_A$ to denote the number of sample pairs that $A$ fails to separate. That is
\[
D_A := \left| \{ (y_i, z_i) \mid \text{$y_i$ and $z_i$ have the same coordinates in $A$}\} \right|~.
\]

If $D_A <  \frac{K k \log m}{10 \epsilon^2}$,  output ``small''.  Otherwise,  output 
\[
\hat{\Gamma}_A = D_A \cdot \frac{\alpha \epsilon^2 {n \choose 2}}{K k \log m}
\] as an estimate of $\Gamma_A$. If $A$ fails to separate at most $(\alpha/100){n \choose 2}$  pairs of tuples, appealing to Chernoff bound yields:
\[
\prob{D_A \geq  \frac{K k \log m}{10 \epsilon^2}} \leq  2 e^{- 4.5 K  k \log m} \leq m^{-100k}~.
\]

If $A$ fails to separate at least $(\alpha/100){n \choose 2}$ pairs, we have
\[
\prob{D_A \notin (1\pm \eps) \Gamma_A  \frac{K k \log m}{\alpha \epsilon^2 {n \choose 2}} } \leq 2 \exp{- \eps^2 \frac{Kk \log m}{\alpha \eps^2 {n \choose 2}} \Gamma_A} \leq m^{-100k}~.
\]
The success probability is therefore at least $1-m^{-90k}$  by appealing to a union bound over at most ${m \choose 1} + {m \choose 2} + \ldots +{m \choose k} \leq m^{k+1}$ possible choices of $A$.

 \subsection{Lower bound}  We show that even for constant $\alpha$,  a valid data sketch must use $\Omega(mk\log{1 / \epsilon})$ bits.  The proof is based on an encoding argument. Let $d_H$ denote the Hamming distance. We will make use of the following communication problem.
 
\begin{lemma}\label{lem:reconstruction-complexity_n3}
	Suppose $Alice$ has a bit string $C$ of length $ktm$ indexed as a $[kt] \times [m]$ matrix. Each of $m$ columns has exactly $k$ one-entries and $k(t-1)$ zero-entries. For Bob to compute a reconstruction $\hat{C}$ of $C$ such that $d_H(C, \hat{C}) \leq \frac{|C|}{10t}$ with probability at least $2/3$, the one-way (randomized) communication complexity is $\Omega(km\log(t))$ bits.
\end{lemma}

The proof of the above lemma is an adaptation of the proof of the Index problem \cite{Ablayev96, KremerNR01} which we defer to the end of this section. We will later set $t = \Theta(1/\sqrt{\epsilon})$ to obtain the lower bound $\Omega(km\log{(1 / \epsilon)})$.

Let $n = kt $ and $D$ be a $n \times m$ matrix whose entries are all ones. Furthermore, let $\mbf{1}_i$ be the $2n \times 1$ canonical vector with the one-entry in the $i$th row. Alice constructs the following $2n \times (m+n)$ data set $M$ as below. Here, we make the assumption that $m \geq \frac{k}{\sqrt{\eps}}$; therefore, $m + n = O(m)$.

 \begin{align*}
 	M = \left[ \begin{array}{ccccccc}
 		C &  & \vertbar & \vertbar &  \ldots     & \vertbar  \\
 		&  &\mbf{1}_1 & \mbf{1}_2 & \ldots & \mbf{1}_n \\ 
 		D &  &  \vertbar & \vertbar &   \ldots  & \vertbar
 	\end{array} \right] =
 	\left[ \begin{array}{c|cccccc}
 		& 1 & 0 & \ldots & 0 \\
 		C & 0 & 1 & \ldots & 0 \\ 
 		& \ldots & \ldots & \ldots   & \ldots \\ 
 		& 0 & 0 & \ldots   & 1\\
 		\hline
 		&  0 & 0 & \ldots &  0 \\
 		D &  0 & 0 & \ldots &  0 \\
 		&  \ldots & \ldots & \ldots &  \ldots \\
 		&  0 & 0 & \ldots &  0 \\
 	\end{array} \right]
 \end{align*}
 
 We will show that Bob can recover $C$, column-by-column using on the described data sketch for non-separation estimation with $\alpha  = 1/16$. Then, appealing to Lemma \ref{lem:reconstruction-complexity_n3}, such a data sketch must have size $\Omega \paren{ mk \log \frac{1}{\eps}}$ bits.

 Fix a column $c$ of $C$. Bob can recover column $c$ of $C$ as follows. If Bob guesses that rows 
 \[R = \{r_1,\ldots,r_{k} \} \subset [kt]
 \]
 in $c$ contain all the one-entries,  Bob can verify if this guess is good using the estimate $\hat{\Gamma}_A$ where 
 \[
 A = \{i, m+r_1,m+r_2,\ldots,m+{r_{k}} \} .
 \]
 
 Here, a guess is good if the reconstruction $\hat{c}$ of $c$ corresponding to this guess satisfies $d_H(\hat{c},c) \leq |c|/(10t)$.
 
 Note that $|A| = k+1=O(k)$. Let $u$ and $v$ be the number of rows in $R$ that Bob guesses correctly and incorrectly respectively. That is $u := |\{ r_i : c_{r_i} = 1 \}|$, and 
 $v := |\{ r_i: c_{r_i} = 0 \}|$. Note that $v = k -u$. The intuition is that for each $r_i$ that is a correct guess, $A$ includes another coordinate that separates another $\approx n+k$ pairs; if $r_i$ is a wrong guess, $A$ includes another coordinate that separates $\approx n-k$ pairs.
 
 \begin{lemma}We have the following equality regarding the number of unseparated pairs in $A$:
 	\[
 	\Gamma_A = (t^2 - t + 5/2)k^2 - (t - 1/2)k + u^2 - 3ku~.
 	\]
 \end{lemma}
 \begin{proof}
 	Consider the following process. Starting with $A$ as the coordinate corresponding to column $c$. We then add coordinates $m+r_i$ for each $i$ to $A$ one by one. Originally, $A$ fails to separate $n+k \choose 2$ pairs corresponding to $n+k$ rows that are 1's in column $c$ and  $n-k \choose 2$ pairs that corresponds to $k$ rows that are 0's in column $c$.
 	
 	For the first correct $r_{i_1}$, $A$ includes a coordinate that separates row $r_{i_1}$ from $n + k - 1$ other rows; for the second correct $r_{i_2}$, $A$ includes a coordinate that separates row $r_{i_2}$ from $n + k-2$ other rows and so on. Similarly, for the first incorrect $r_{j_1}$, $A$ includes a coordinate that separates row $r_{j_1}$ from $n - k - 1$ other rows; for the second incorrect $r_{j_2}$, $A$ includes a coordinate that separates row $r_{j_2}$ from $n - k-2$ other rows and so on.
 	
 	Therefore, in the end, the number of pairs that remain unseparated in $A$ is
 	\begin{align}
 		& {n + k \choose 2} + {n - k \choose 2} - \sum_{a=1}^u (n + k - a) - \sum_{b=1}^v (n - k - b) \nonumber \\
 		 = & {n + k \choose 2} + {n - k \choose 2} - u(n+k) - (k-u)(n-k) + \frac{u(u+1)}{2}   + \frac{(k-u)(k-u+1)}{2}  \nonumber  \\
 		=&  {n + k \choose 2} + {n - k \choose 2} + u^2 - 3ku + \frac{k^2+k}{2} - k(n-k) \nonumber  \\
 		=&  \frac{(n+k)(n+k-1)}{2} + \frac{(n-k)(n-k-1)}{2} + u^2 - 3ku + \frac{k^2+k}{2}   - k(n-k) \nonumber  \\
 		=& n^2 + k^2 - n  + u^2 - 3ku + \frac{k^2+k}{2} - k(n-k) + u^2 - 3ku \nonumber \\
 		=& n^2 + \frac{5k^2}{2} - n - kn + \frac{k}{2} + u^2 - 3ku \nonumber \\
 		 =&  (t^2 - t + 5/2)k^2 - (t - 1/2)k + u^2 - 3ku \label{eq:sep-count}~. 
 	\end{align} 
 \end{proof}

First, note that $\Gamma_A > {n \choose 2} > \frac{1}{16} {2n \choose 2}$. Thus, a correct data sketch must output the required estimate  $\hat{\Gamma}_A = (1\pm \eps)\Gamma_A$.
 
 Observe that the expression \eqref{eq:sep-count} is decreasing for $u \leq 3k/2$. Suppose $u \leq 0.9k$, 
 \begin{align*}
 	&\Gamma_A \geq (t^2 - t + 5/2 + 0.9^2 - 2.7)k^2 - (t - 1/2)k \\
 	=& (t^2 - t + 0.61)k^2 - (t - 1/2)k~.
 \end{align*}
 
  On the other hand, if $u = k$, then 
  \begin{align*}
 	\Gamma_A  &\leq (t^2 - t + 5/2- 2)k^2 - (t - 1/2)k\\ 
 	&= (t^2 - t + 0.5)k^2 - (t - 1/2)k~.
 \end{align*}

Our objective is to choose $t$ such that if the provided estimate $\hat{\Gamma}_A = (1 \pm \epsilon) \Gamma_A$, Bob can tell whether his guess is good.  To do so, it suffices to choose $t$ such that:

 \begin{align*}
 	\frac{t^2 - t + 0.61}{t^2 - t + 0.5} > \frac{1 + \epsilon}{1 - \epsilon} \iff \frac{11}{200t^2 -200t + 11} > \eps ~.
 \end{align*}

Hence, we can set $t = \frac{1}{K\sqrt{\epsilon}}$ for some sufficiently large constant $K$. Specifically, Bob queries the estimates $\hat{\Gamma}_A$ for  at most $n \choose k+1$ choices of $A$. If $\hat{\Gamma}_A \leq (1+\eps)((t^2 - t + 0.5)k^2 - (t - 1/2)k)$, then he knows that the guess is good; he will use the corresponding reconstruction for column $c$.

Therefore, provided a valid sketch. Bob can correctly compute a  reconstruction $\hat{C}$ of $C$ such that
\[
d_H(C, \hat{C}) \leq \frac{|C|}{10 t}~.
\]
with probability at least 3/4. Reparameterize $k \leftarrow k+1$ and $m \leftarrow m +n $ gives us the lower bound.
 
 \begin{proof}[Proof of Lemma \ref{lem:reconstruction-complexity_n3}]
For convenience, let $N = mkt$ (i.e., the length of the bit string $C$ that is given to Alice) and $L = mk \log t$ (we want to show that $\Omega(L)$ communication is needed). 

Recall that Bob wants to compute a reconstruction $\hat{C}$ of $C$ such that 
\[
d_H(\hat{C},C) \leq \frac{|C|}{10t}~.
\] 
We consider the distributional complexity methodology where $C$ follows a distribution $\mathcal{D}$ and communication protocols are deterministic. Here, we simply choose $\mathcal{D}$ in which in each column, $k$ bits are chosen uniformly at random to be 1's and the remaining $(t-1)k$ bits are set to 0's. It suffices to show that for any deterministic protocol $P$ that uses $0.001L$ bits of communication,
\[
\probb{C \sim \mathcal{D}}{\text{$P$ fails to reconstruct $C$ as required} } > 1/3~.
\]

By the minimax principle  \cite{Yao83}, this implies that there exists an input $C$ such that any randomized protocol that uses $0.001 L $ bits of communication will fail to reconstruct  $C$ as required with probability more than 1/3.

Suppose Bob gets a message $Z$ from Alice ($Z$ is a function of $C$). Let $Q(Z)$ be his reconstruction based on $Z$. Let 
\[
\mathcal{A} = \{ Q(Z): Z \in \{0,1\}^{N} \}
\] be the set of all possible reconstructions by Bob. Note that $|\mathcal{A}| \leq 2^{0.001L}$. 

We say Alice's input $C$ good if there exists a reconstruction $Q$ such that 
\[
d_H (Q, C) \leq \frac{N}{10t}~.
\]

Otherwise, we say $C$ is bad. Fix a reconstruction $Q$, the number of inputs $C$ whose Hamming distance is at most $\frac{N}{10t}$ from $Q$ is
\begin{align}
{N \choose 0} + {N \choose 1} + {N \choose 2} + \ldots + {N \choose N/(10t)} \label{eq:hamming-1}~.
\end{align}
Recall that $(a/b)^b \leq {a \choose b} \leq (ea/b)^b$. For sufficiently large $k,m,$ and $t$, expression \eqref{eq:hamming-1} can be upper bounded as:
\begin{align*}
N (10 t e)^{ N/(10t)} &= N (10 t e)^{0.1km} \\
&= N 2^{0.1 km \log (10et) }\\ 
&\leq  N 2^{0.1 km \log t + 0.1 km \log (10e)}\\
&< N 2^{ 0.11 km \log t}~.
\end{align*}
Hence, for sufficiently large $k,m,$ and $t$, the total number of good inputs is at most
\begin{align*}
& |\mathcal{A}|  N 2^{ 0.11 km \log t} \leq  2^{\log N +0.111 km \log t  } \\
=&2^{\log k + \log m + \log t +0.111 km \log t  } < 2^{0.2 km \log t  }.
\end{align*}
Therefore,
\begin{align*}
& \probb{C \sim \mathcal{D}}{\text{$P$ fails to reconstruct of $C$ as required}  }  \\
 \geq  & \probb{C \sim \mathcal{D}}{\text{$P$ fails to reconstruct $C$ as required} \given \text{$C$ is bad}} \times \probb{C \sim \mathcal{D}}{\text{$C$ is bad}} \\
 \geq  &  \left(  1-  \frac{2^{0.2k m \log t}}{{tk \choose k}^m} \right) \times 1 \geq 1-\frac{t^{0.2km}}{t^{km}}= 1 - \frac{1}{t^{0.8km}} > 1/3. \qedhere
\end{align*}
\end{proof}

 \section{Implementation} \label{sec:implementation}

The modified algorithm to detect $\epsilon$-separation keys is very simple, despite the involved analysis. Therefore, we briefly demonstrate the effectiveness of our approach described in Section \ref{sec:improved-upperbound} and compare it to the naive approach given by Motwani and Xu  \cite{MX07}. We ran our experiments on an M1 Pro processor with 16 gigabytes of unified memory. The code for the experiments is available on GitHub \cite{Github}.

\paragraph{Description of datasets.} We used two of the data sets that were tested in \cite{MX07}, the adult income data set and the covtype data set, both of which are from the UCI Machine Learning Repository. We also used the 2016 Current Population Survey, which is publicly provided by the US Census. These data sets are a good representation in terms of data size and attribute size. For example, the adult data set contains slightly more than 32,000 values with 14 attributes, while the census data contains millions of records with 388 attributes.

\paragraph{Comparison methodology. } We compare the two approaches with $\epsilon = 0.001$ and $\delta=0.01$. These are the same tuning parameters that were considered in \cite{MX07}. We compare our results in terms of the following: (i) sample size, (ii) run-time , and (iii) the percentage in which both algorithms agree on accepting or rejecting a set of attributes. Note that in some cases, even though the two algorithms' outputs are different, both can be correct. In particular, if a subset is not a perfect key but it separates at least $(1-\epsilon){n \choose 2}$ pairs, then it is correct to either accept or reject.

For each data set, we select about $100$ random subsets of attributes to query. See Table \ref{experiments_table} for the detailed results. At a high level, both approaches agree on nearly all queries while requiring a much smaller sample size.

\begin{table}[h]
    \caption{Sample size and average running time across 10 different trials, $\star$ and $\star \star$ denote the results by the approaches in \cite{MX07} and in this paper respectively. S: sample sizes. T: time. A: agreement}
    \label{experiments_table}
    \begin{tabular}{l l l l l l}
    \hline
    Dataset & S ($\star$) & S ($\star \star$) & T  ($\star$) &  T  ($\star \star$) & A \%  \\ [0.8ex] 
    \hline
    Adult & 13,000 & 411 & 1.903 sec & 0.208 sec & 95\% \\ 
    Covtype & 55,000 & 1,739 & 188.02 sec & 2.49 sec & 98\% \\
    CPS & 372,000 & 11,764 & 790.08 sec & 60.03 sec& 100\% \\
    \hline
     \label{tb:1}
    \end{tabular}
\end{table}

\bibliographystyle{plain}
\bibliography{references}

\begin{thebibliography}{10}

\bibitem{Ablayev96}
Farid~M. Ablayev.
\newblock Lower bounds for one-way probabilistic communication complexity and
  their application to space complexity.
\newblock {\em Theor. Comput. Sci.}, 157(2):139--159, 1996.

\bibitem{AnanthakrishnaCG02}
Rohit Ananthakrishna, Surajit Chaudhuri, and Venkatesh Ganti.
\newblock Eliminating fuzzy duplicates in data warehouses.
\newblock In {\em {VLDB}}, pages 586--597. Morgan Kaufmann, 2002.

\bibitem{ChaudhuriGGM03}
Surajit Chaudhuri, Kris Ganjam, Venkatesh Ganti, and Rajeev Motwani.
\newblock Robust and efficient fuzzy match for online data cleaning.
\newblock In {\em {SIGMOD} Conference}, pages 313--324. {ACM}, 2003.

\bibitem{CormodeDW21}
Graham Cormode, Charlie Dickens, and David~P. Woodruff.
\newblock Subspace exploration: Bounds on projected frequency estimation.
\newblock In {\em {PODS}}, pages 273--284. {ACM}, 2021.

\bibitem{FriedmanS22}
Roy Friedman and Rana Shahout.
\newblock Box queries over multi-dimensional streams.
\newblock {\em Inf. Syst.}, 109:102086, 2022.

\bibitem{giannella2002using}
Chris~M Giannella, Mehmet~M Dalkilic, Dennis~P Groth, and Edward~L Robertson.
\newblock Using horizontal-vertical decompositions to improve query evaluation.
\newblock In {\em Proceedings of the 19th British National Conference on
  Databases (BNCOD), Lee. Notes in Comp. Sci. vol}, volume 2405, pages 26--41.
  Citeseer, 2002.

\bibitem{Github}
Ryan Hildebrant.
\newblock Github.
\newblock \url{https://github.com/Ryanhilde/min_set_cover/}.

\bibitem{KivinenM92}
Jyrki Kivinen and Heikki Mannila.
\newblock Approximate dependency inference from relations.
\newblock In {\em {ICDT}}, volume 646 of {\em Lecture Notes in Computer
  Science}, pages 86--98. Springer, 1992.

\bibitem{KremerNR01}
Ilan Kremer, Noam Nisan, and Dana Ron.
\newblock Errata for: "on randomized one-round communication complexity".
\newblock {\em Comput. Complex.}, 10(4):314--315, 2001.

\bibitem{KvetonMVX18}
Branislav Kveton, S.~Muthukrishnan, Hoa~T. Vu, and Yikun Xian.
\newblock Finding subcube heavy hitters in analytics data streams.
\newblock In {\em {WWW}}, pages 1705--1714. {ACM}, 2018.

\bibitem{LibertyMTU16}
Edo Liberty, Michael Mitzenmacher, Justin Thaler, and Jonathan~R. Ullman.
\newblock Space lower bounds for itemset frequency sketches.
\newblock In {\em {PODS}}, pages 441--454. {ACM}, 2016.

\bibitem{mcconnell2001inequality}
Terry~R McConnell.
\newblock An inequality related to the birthday problem.
\newblock {\em Preprint}, 2001.

\bibitem{motwani1995randomized}
Rajeev Motwani and Prabhakar Raghavan.
\newblock {\em Randomized algorithms}.
\newblock Cambridge university press, 1995.

\bibitem{MX07}
Rajeev Motwani and Ying Xu.
\newblock Efficient algorithms for masking and finding quasi-identifiers.
\newblock In {\em Technical Report}, 2008.

\bibitem{NoceWrig06}
Jorge Nocedal and Stephen~J. Wright.
\newblock {\em Numerical Optimization}.
\newblock Springer, New York, NY, USA, 2e edition, 2006.

\bibitem{PfahringerK95}
Bernhard Pfahringer and Stefan Kramer.
\newblock Compression-based evaluation of partial determinations.
\newblock In {\em {KDD}}, pages 234--239. {AAAI} Press, 1995.

\bibitem{steele2004cauchy}
J~Michael Steele.
\newblock {\em The Cauchy-Schwarz master class: an introduction to the art of
  mathematical inequalities}.
\newblock Cambridge University Press, 2004.

\bibitem{wiki:Chernoff_bound}
Wikipedia.
\newblock {Chernoff bound} --- {W}ikipedia{,} the free encyclopedia.
\newblock
  \url{http://en.wikipedia.org/w/index.php?title=Chernoff\%20bound&oldid=1119845299},
  2022.
\newblock [Online; accessed 18-November-2022].

\bibitem{Yao83}
Andrew~Chi{-}Chih Yao.
\newblock Lower bounds by probabilistic arguments (extended abstract).
\newblock In {\em {FOCS}}, pages 420--428. {IEEE} Computer Society, 1983.

\bibitem{Young08}
Neal~E. Young.
\newblock Greedy set-cover algorithms.
\newblock In Ming-Yang Kao, editor, {\em Encyclopedia of Algorithms}. Springer,
  2008.

\end{thebibliography}

\appendix

\section{Applying Karush–Kuhn–Tucker (KKT)  conditions} \label{sec:kkt}
 
In this section, we present the KKT condition for self-containedness. This classical result is presented as in \cite[Chapter $12$]{NoceWrig06}. Consider a constrained optimization problem with an objective function $f$, inequality and equality constraints $c_i, i \in \mc{I} \cup \mc{E}$. 
\begin{equation}
	\label{eq:optimprob}
	\begin{aligned}
		\text{Maximize} & ~~~~ f(s) \\
		\text{Subject to:} \\
		c_i(s) & \geq 0 & \text{for $i \in \mc{I}$~.}\\
		c_i(s) & = 0 & \text{for $i \in \mc{E}$~.}
	\end{aligned}
\end{equation}

\begin{definition}[Active set \cite{NoceWrig06}]
	\label{def:activeset}
	The active set $\mc{A}(s)$ of a feasible solution $s$ of \eqref{eq:optimprob} is a set of indices defined as: $\mc{A}(s) := \{i \in \mc{I} \mid c_i(s) = 0\} \cup \mc{E}$. In other words, $\mc{A}(s)$ refers to the set of equality and inequality (where equality happens) constraints.
\end{definition}

\begin{definition}[LICQ \cite{NoceWrig06}]
	\label{def:licq}
	Given a feasible solution $s$ of \eqref{eq:optimprob} and its active set in Definition \ref{def:activeset}, we say that \textit{linear independence constraint qualification} (LICQ) holds at $s$ if $\{\nabla c_i(s), i \in \mc{A}(s)\}$ is linearly independent.
\end{definition}

\begin{theorem}[KKT conditions {\cite{NoceWrig06}}]
	\label{theorem:KKT}
	Consider $s^\star$ a local minimum of \eqref{eq:optimprob}, whose $f$ and $c_i, i \in \mc{I} \cup \mc{E}$ are continuously differentiable and that the LICQ holds at $s^\star$, then there is a Lagrange multiplier vector $\lambda^\star$ with components $\lambda_i^\star, i \in \mc{I} \cup \mc{E}$ such that the following conditions are satisfied:
	\begin{align}
		c_i(s^\star) & \geq 0, \forall i \in \mc{I} && \text{(Primal feasibility for inequality)} \\
		c_i(s^\star) & =  0, \forall i \in \mc{E} && \text{(Primal feasibility for equality)} \\
		\lambda^\star_i & \geq 0, \forall i \in \mc{I} && \text{(Dual feasibility)} \\
		\nabla f(s^\star) & =  \sum_{i=1}^t \lambda^\star_i \nabla g^\star_i(s) &&\text{(Stationarity)} \\
		\lambda^\star_i c_i(s^\star) & = 0, \forall i \in \mc{E} \cup \mc{I} && \text{(Complementary slackness)}~.
	\end{align}
\end{theorem}
\paragraph{Application to our problem.} Recall that we want to maximize:

\begin{align*}
f(s) =  \sum_{1 \leq j_1 < j_2 < \ldots < j_r \leq n} s_{j_1} s_{j_2} \ldots s_{j_r}  ~.
\end{align*}

Subject to:
\begin{align*}
\sum_{i=1}^n s_i & = n \\
s_i & \geq 0 \text{, for $i=1,2,\ldots,n$} \\
\sum_{i=1}^n s_i^2 & \geq \epsilon n^2/4 \text{, for $i=1,2,\ldots,n$.}
\end{align*}

KKT conditions imply that when LICQ holds, there exist constants $\eta \in \mathbb{R}, \{\lambda_i \}_{i \in [n]} \geq 0, \mu \geq 0$ such that for any local maxima, we have:
\begin{align*}
\nabla f(s) = \mu \nabla \left(\sum_{i=1}^n s_i^2 \right) +    \sum_{i=1}^n \lambda_i \nabla  (s_i) + \eta  \nabla \left(\sum_{i=1}^n s_i \right) ~.
\end{align*}
Furthermore, according to the complementary slackness condition, for each $i=1,2,\ldots,n$:
\begin{align*}
\lambda_i s_i = 0
\end{align*}
Therefore, if $s_i > 0$ then $\lambda_i=0$. 

\section{Improvements for Approximate Minimum $\eps$-Separation Key} \label{sec:min-approx-key-improvements}

In this section, we prove Proposition \ref{prop:approximate-min-key}. Recall from Section \ref{sec:intro} that we treat ${R \choose 2}$ as the ground set and each coordinate is a set that contains all pairs that it separates. The minimum key $I^\star$ separates all pairs in ${R \choose 2}$ since ${R \choose 2} \subset {X \choose 2}$. Furthermore, in Section \ref{sec:improved-upperbound}, we showed that no subset of coordinates that separates fewer than $(1-\eps){n \choose 2}$ pairs is a set cover of ${R \choose 2}$ with high probability. Hence, a $\gamma$ approximation to this minimum set cover instance yields an $\eps$-separation key whose size is at most $\gamma | I^\star|$ with high probability. This implies the improvement in terms of sample size.

In terms of running time, to achieve the approximation factor $\gamma = O\paren{ \ln \frac{m}{\eps}}$, both the algorithm in \cite{MX07} and ours use the greedy set cover algorithm \cite{Young08} (described in Algorithm \ref{algorithm1}). 

Let $A$ be the output, initialized to $\emptyset$. The greedy algorithm, at each step, adds the coordinate that separates the most number of currently unseparated pairs to $A$. The algorithm stops when all pairs are separated by $A$. 

The key difference between two algorithms is: ours samples $R = \Theta(m / \sqrt{\epsilon})$ tuples and use $R \choose 2$ as ground set while the approach of \cite{MX07} samples $R' = \Theta(m/ \epsilon)$ {pairs of tuples} of the data set and use $R'$ as the ground set. Algorithm \ref{algorithm1} runs in $O\left(M^2N\right)$  time where $N$ is the ground set's size and $M$ is the number of sets in the input. Therefore, our algorithm and that of Mowani and Xu \cite{MX07} yield time complexity $O\paren{ \frac{m^4}{\epsilon} }$ and $O \left( \frac{m^3}{\eps} \right)$ respectively. However, our algorithm can be refined to a better time complexity $O \left( \frac{m^3}{\sqrt{\eps}} \right)$. We outline how to achieve this running time as follows.

\begin{algorithm}[H]
	\centering
	\caption{Greedy Set Cover Algorithm} 
	\label{algorithm1}
	\begin{algorithmic}[1]
		\Require $X = \{X_1, \ldots, X_N\}$ ground set, $S = \{S_1, \ldots, S_M\}$ collections of subsets of $X$. 
		\State $U = X$ \Comment{$U$ stores the uncovered elements}
		\State $C = \emptyset$ \Comment{$C$ stores the indices of sets of the cover}
		\While{$U \neq \emptyset$}
			\State Select $S_i \in S$ that maximizes $|S_i \cap U|$.
			\State $U \leftarrow U \setminus S_i$.
			\State $C \leftarrow C \cup \{i\}$.
		\EndWhile
		\State \Return $C$	
	\end{algorithmic}
\end{algorithm}

We can visualize this process in terms of auxiliary graphs introduced in Section \ref{sec:improved-upperbound}. Let $A$ be the output, initialized to $\emptyset$. Originally, $G_A$ consists of one clique $C$ that contains all $x_j \in R$ since the empty set does not separate any pair. 

As we add more coordinates to the solution, $A$ will separate more pairs in ${R \choose 2}$ and the number of disjoint cliques $c$ in $G_A$ increases. We stop when $c = |R|$; in other words, all pairs in ${R \choose 2}$ are separated by $A$. We will make use of the following procedure.

\paragraph{Partitioning.} Let $C \subseteq R$ be a subset of the sampled tuples. For some given $i$, we can partition $C$ into $D_1,D_2,\ldots$ based on the $i$th coordinates. The simplest approach is to sort the $i$th coordinates of the tuples in $C$. This takes $O\paren{|C|\log |C|}$ time.

Let $C_{y,1}, C_{y,2}, \ldots,$ be the disjoint cliques in $G_A$ after $y$ steps (i.e., after we have added $y$ coordinates to $A$). Originally, $C_{1,1} = R$. 

At each step $y$, for each coordinate $k \notin A$, adding $k$ to $A$ breaks each clique $C_{y,i}$  into some new cliques $D_1^{(i)}, D_2^{(i)},\ldots$ based on the $k$th coordinates.  This corresponds to the fact that adding $k$ to $A$ results in $A$ separating more pairs of tuples.  We use the procedure $D^{(i)} \leftarrow Partition(C_{y,i},k)$ to compute $D_1^{(i)}, D_2^{(i)},\ldots$ We observe that coordinate $k$ would separate

\begin{align*}
	g_k := \frac{1}{2} \sum_i \sum_{ a,b} {|D_a^{(i)}| |D_b^{(i)}|} & =  \frac{1}{2}  \sum_i \left[ \left( \sum_{a} |D_a^{(i)}|  \right)^2 -  \sum_{a} |D_a^{(i)}|^2    \right] \\
	&=  \frac{1}{2}  \sum_i \left[  |C_i|^2 -  \sum_{a} |D_a^{(i)}|^2    \right] ~.
\end{align*}
new pairs of tuples in ${R \choose 2}$.

Given $D^{(i)}$, $g_k$ can be computed in $O(\sum_{i} |C_i| ) = O(|R|)$ time by computing $\sum_{a} |D_a^{(i)}|^2$ as there can be at most $|C_i|$ terms in the sum.  The algorithm then adds coordinate $k$ with the highest $g_k$ to $A$ and update the cliques accordingly. In other words, replace $C_i$ with $D_1^{(i)}, D_2^{(i)},\ldots$
For each $k \notin A$, the time to partition based on the $k$th coordinates is
\[
O \left( \sum_i |C_i | \log |C_i |  \right) = O(|R| \log |R|) ~.
\]
We need to do this for at most $m$ coordinates not in $A$ and the process repeats for at most $m$ steps. Thus, the algorithm's running time is
\[
O \left( m^2  |R| \log |R| \right) = O \left( \frac{m^3}{\sqrt{\eps}} \log \frac{m}{\eps} \right)~.
\]

If we allow an extra factor $m$ in the space use, we can reduce the running time further by improving the partitioning procedure. For each $k \notin A$, the partition step can be done can be done in $ O \left( \sum_i |C_i |  \right) = O(|R| )$ time  (instead of $O(|R| \log |R|)$ time) using Algorithm \ref{algorithm2}. 

Algorithm \ref{algorithm2} is provided with a lookup table $P \in \mathbb{N}^{|R| \times m}$ and an index $k$.  We partition $R$ based on the $k$th coordinates into $C^{(k)}_1$, $C^{(k)}_2$, $\ldots$. The value $P[k,j]$ is equal to the index of the partition that $x_j \in R$ belongs to. In other words, $x_j \in C^{(k)}_{P[k,j]}$. Pre-computing $P$ involves sorting each of $m$ coordinates of tuples in $R$. The running time to compute $P$ is therefore
\[
O(m|R|\log |R|) = O\paren{\frac{m^2}{\sqrt{\epsilon}}\log\left(\frac{m}{\epsilon}\right)} ~.
\]

Thus, the overall running time is
\begin{equation*}
	O\paren{\frac{m^2}{\sqrt{\epsilon}}\log\left(\frac{m}{\epsilon}\right)} + O\paren{\frac{m^3}{\sqrt{\epsilon}}} = O\paren{\frac{m^3}{\sqrt{\epsilon}}}~.
\end{equation*}

\begin{algorithm}[H]
	\centering
	\caption{Partitioning $C$ based on the $k$th coordinates using the look-up table $P$} 
	\label{algorithm2}
	\begin{algorithmic}[1]
		\Require A subset $C \subseteq R$, a look-up table $P$ and an index $1 \leq k \leq m$.
		\State Initialize $D$ as an array of empty list. \Comment{$D$ contains $D^{(i)}$ as an array of list.}
		\State Initialize $L$ as an empty list. \Comment{$L$ is the list of indices $j$ s.t $D[j]$ is non empty}
		\For{$x_j \in C$}
		\If{$D[P[k,j]] = \emptyset$}
			\State $L = L \cup \{P[k,j]\}$
		\EndIf
		\State $D[P[k,j]] \leftarrow D[P[k,j]] \cup \{x_j\}$.
		\EndFor
		\State \Return $\{D[j] \mid j \in L\}$.
	\end{algorithmic}
\end{algorithm}

\section{Omitted Proofs and Examples} \label{sec:omitted-proofs}
\subsection{Proof of Lemma \ref{lem:lbfixeddelta}}
\label{subse:lemlbfixeddelta}
\begin{proof}
	For technical reason, we will consider $\epsilon$ such that $1 / \epsilon = q + 1 / 2$ where $q \in \mathbb{N}$. Our construction of the data set $\mc{D}$ is $\mc{D} := \{1, 2, \ldots, q\}^m = [q]^m$. Thus, $n: = |\mc{D}| = q^m$. We choose $m$ such that $m < \exp{\frac{1}{4}(\frac{1}{\epsilon} - \frac{1}{2})}$ (or equivalently $\log m < q/4$). 
	
	Firstly, we will prove that all subsets of attributes of cardinality one is bad. Indeed, given any singleton set $A = \{i \mid 1 \leq i \leq m\}$, the auxiliary graph $G_A = (V_A, E_A)$ will be decomposed evenly into $q$ cliques whose size are $n / q$. The number of edges (of $G_A$) is:
	\begin{align*}
		|E_A| = q \frac{n/q (n / q - 1)}{2} = \frac{n(n/q - 1)}{2}~.
	\end{align*}
	To show $|E_A| > \epsilon n(n-1)/2$, it is sufficient to demonstrate that:
	\begin{equation*}
		\begin{aligned}
			n/q - 1 > \epsilon(n-1) & \iff \frac{n}{\epsilon q} - \frac{1}{\epsilon} > n - 1\\
			& \iff n \frac{q + 1/2}{q} - q - \frac{1}{2} > n - 1\\
			&\iff  \frac{n}{2q} > q - \frac{1}{2} \\
			&\iff  q^{m-1} > 2q - 1 \\ 
		\end{aligned}
	\end{equation*}

	which is true if $m \geq 3, q > 1$. 
	
	Denote $R_{A,r}$ the event that a bad subset $A$ is detected after $r$ samples, we derive our lower bound by bounding $\prob{\bigcap_{A : |A| = 1} R_{A,r}}$, which is the probability that one can detect all the bad subsets of cardinality one. To distinguish between the case of sampling with and without replacement, we denote $\probw{E}$ and $\probo{E}$ the probability of an event $E$ under the sampling with and without replacement respectively.
	
	We first deal with the case of uniform sampling with replacement. In this case, sampling a tuple $x$ is equivalent to sampling each coordinate i.i.d from uniform multinomial distribution of the set $\{1, \ldots, q\}$. Hence,
	\begin{align*}
		\probw{\bigcap_{A : |A| = 1} R_{A,r}} = \prod_{i =1}^m \probw{R_{\{i\},r}} 	= \probw{R_{\{i\},r}}^m~.
	\end{align*} 
	Moreover, we have: $\probw{R_{A,r}} \leq 1 - \prod_{i = 0}^{r-1} (1 - i / q)$. This is not an equality since we need to exclude the cases where a tuple of $\mc{D}$ is sampled more than once. Since $\exp{-2x} \leq 1 - x, \forall x \leq 1/2$, we have:
	\begin{align*}
		\probw{R_{A,r}} \leq 1 - \exp{-r(r-1) / q} \leq 1 - \exp{-r^2 / q}~.
	\end{align*}
	for $r \leq q / 2$. If one chooses $r = \sqrt{q\log m}$ (which makes the inequality $\exp{-2x} \leq 1 - x$ valid since $\log m \leq q/4$), we have:
	\begin{equation*}
		\probw{R_{A,r}} \leq 1 - \frac{1}{m}
	\end{equation*}
	Finally, one can bound $\probw{\bigcap_{A \mid |A| = 1} R_{A,r}} = \probw{R_{A,r}}^m \leq (1 - 1/m)^m \leq 1/e$. Thus, the sampling complexity is lower bounded by $\Theta(\sqrt{q\log m}) = \Theta\left(\sqrt{\frac{\log m}{\epsilon}}\right)$.
	
	For the case of sampling without replacement, we will use the following observation: Let $x := (x_1, \ldots, x_r)$ be a sequence of distinct elements in $\mc{D}$ and $X_i, 1 \leq i \leq r$ be the $i$th element sampling from $\mc{D}$, we have:  	
	\begin{align*}
		& \probo{X_{i} = x_i, \forall 1 \leq i \leq r} = \frac{1}{n(n-1)\ldots(n-r+1)} \\
		& \probw{X_{i} = x_i, \forall 1 \leq i \leq r} = \frac{1}{n^r} 
	\end{align*}
	Let $\mc{X}$ be the set of all sequences of distinct tuples of length $r$ that remove all bad subsets and $R_r$ be the event of all bad subsets being detected (i.e., there exists an unseparated pair in the sample for each bad subset) after sampling $r$ samples, we have: 
	
	\begin{equation*}
		\begin{aligned}
			\probo{R} &= \sum_{x \in \mc{X}} \probo{X = x}\\
			&= \frac{n^s}{n(n-1)\ldots(n-s+1)}\sum_{x \in \mc{X}} \probw{X = x}\\ 
			&= \frac{n^s}{n(n-1)\ldots(n-s+1)} \probw{R}\\
			&\overset{\eqref{eq:boundedterm}}{\leq} e^{s(s-1)/(n-s+1)}\probw{R}.
		\end{aligned}
	\end{equation*}
	We choose $m$ such that $n - s + 1 \geq \frac{1}{\ln 2} s(s-1)$. With $s = \sqrt{q\log m}$ (as in the case of sampling with replacement), this is equivalent to:
	\begin{equation*}
		\begin{aligned}
			q^m \geq \frac{1}{\ln 2} \sqrt{q\log m}(\sqrt{q\log m} + 1) + \sqrt{q\log m} - 1 \\
		\end{aligned}
	\end{equation*}
	which is true for $q > 3, m > 2$. If $n - s + 1 \geq \frac{1}{\ln 2} s(s-1)$, we have:
	\begin{equation*}
		\probo{R} \leq 2\probw{R}
	\end{equation*}
	Thus, we can derive the same complexity bound $\Theta(\sqrt{q\log m}) = \Theta\left(\sqrt{\frac{\log m}{\epsilon}}\right)$ for the case of sampling without replacement with failure probability $2 / e$. 
\end{proof}

\subsection{Proof of Lemma \ref{lem:lblargedelta}}
\label{subsec:lemlblargedelta}
\begin{proof}
	We choose $n$ such that $n  \gg m^2/\epsilon$. We construct the data set $\mc{D}$ of $n$ tuples $\mc{D}:= \{x_1, \ldots, x_n\}$ as follows.
	\begin{enumerate}[leftmargin=*]
		
		\item The first coordinate of each of $\{x_i\}_{i \in [n]}$ has $q:= (1 - \sqrt{2\epsilon})n + 1$ distinct values. Without loss of generality, we assume the set of distinct values of the first coordinate are $\{1, \ldots, q\}$. More importantly, there are exactly $\sqrt{2\epsilon} n$ tuples having their first coordinates equal to $1$ and $(1 - \sqrt{2\epsilon})n$ remaining tuples having distinct first coordinates in the set $\{2, \ldots, q\}$. We denote $S$ the set of tuples whose first coordinates are $1$ and the set of remaining tuples as $\bar{S} = \{x_i \mid 1 \leq i \leq n\} \setminus S$. Thus, the graph $G_{\{1\}}$ has one big   clique of sizes $\sqrt{2\epsilon} n$ and $(1-\sqrt{2\epsilon})n$ isolated vertices.
		
		\item  Each of the remaining coordinates can be chosen so that a key exists for $\mc{D}$. 
	\end{enumerate}
	
	It is worth noting that $A := \{1\}$ is a bad subset of attributes. This is clear since:
	\begin{equation*}
		|E(G_A)| = \frac{\sqrt{2\epsilon}n(\sqrt{2\epsilon}n - 1)}{2} > \epsilon \frac{n(n-1)}{2}~.
	\end{equation*}
	Assume we samples $r$ times and $r \leq n /2$. The event $\overline{R_{A,r}}$ (i.e., failing to reject $A$) includes the event $\overline{R_{A,r}'}$ where we get exactly $r$ distinct elements of $\bar{S}$. Hence,
	\begin{align*}
		\mathbb{P}(\overline{R_{A,r}}) &\geq \prob{\overline{R'_{A,r}} } \\
		&\geq \frac{(1 - \sqrt{2\epsilon})n}{n} \frac{(1 - \sqrt{2\epsilon})n - 1}{n - 1} \ldots \frac{(1 - \sqrt{2\epsilon})n - r + 1}{n - r + 1}\\
		&= \prod_{i = 0}^{r-1} \left(1 - \sqrt{2\epsilon} - \frac{\sqrt{2\epsilon}i}{n-i}\right)\\
		&\geq \prod_{i = 0}^{r-1} \left(1 - \sqrt{2\epsilon} - \frac{\sqrt{2\epsilon}i}{n-r}\right)\\
		&\geq \prod_{i = 0}^{r-1} \left(1 - \sqrt{2\epsilon} - \frac{2\sqrt{2\epsilon}i}{n}\right)~.
	\end{align*}
	
	Since $e^{-2x} \leq 1 - x$ for all $x \leq 1/2$, we further have:
	\begin{align*}
		& \prob{\overline{R_{A,r}}} \geq \prod_{i = 0}^{r-1} \exp{-2\left(\sqrt{2\epsilon} + \frac{\sqrt{2\epsilon}i}{n-r}\right)} \\
		=& \exp{-2\sqrt{2 \epsilon}r - 2\sqrt{2\epsilon}\frac{r(r-1)}{n}}~.
	\end{align*}
	
	To reject $A$ with probability $1 - \exp{-m}$, one needs to have $\prob{\overline{R_{A,r}}} \leq e^{-m}$. Our analysis implies:
	\begin{equation*}
		m \leq 2r\sqrt{2}\sqrt{\epsilon} + 2\sqrt{2\epsilon}\frac{r(r-1)}{n}
	\end{equation*}
	For $r = \frac{m}{4\sqrt{\epsilon}}$, we have:
	\begin{equation*}
		\begin{aligned}
			m \leq & 2\sqrt{2\epsilon}r + 2\sqrt{2\epsilon}\frac{r(r-1)}{n} &\leq  \frac{m}{\sqrt{2}} + 2\sqrt{2\epsilon} \leq \frac{m}{\sqrt{2}} + 2\sqrt{2} 
		\end{aligned}
	\end{equation*}
	since we choose $n$ such that $r^2 = m^2/\epsilon \ll n$. This does not hold for $m \geq 4 / (\sqrt{2} - 1)$. This shows one needs to sample $\Omega(m / \sqrt{\epsilon})$ to reject $A$ with probability $1 - \exp{-m}$.
\end{proof}

\subsection{An Example Regarding Lemma \ref{lem:optimal-property}}\label{sec:Appendix-examples}

Let $\epsilon' = \epsilon/4$. We provide an example rejecting the intuition that the optimal value $s^\star \in P$ that maximizes the non-collision probability must have equal non-zeroes entries (with values $\epsilon' n$). This implies that Lemma \ref{lem:optimal-property} is necessary.

In particular, let $n=40,\epsilon' = 1/4^2 = 0.0625$, and $r = 10$. Consider $s_1 \in P$ where
\[
s_1 = (\underbrace{2.5,\ldots,2.5}_{16 \text{ times }},0,0,\ldots,0).
\]
Note that $2.5^2 \times 16 = 0.0625 \times 40^2 = \epsilon' n^2$.
 
 Then, consider $s_2 \in P$ where
\[
s_2 = (10, \underbrace{1,\ldots,1}_{30 \text{ times}}, 0, \ldots, 0).
\]
Also note that $10^2 + 30\times 1^2 > 0.0625 \times 40^2 = \epsilon' n^2$ as required. One can check that $f(s_1) \approx 76370239.25\ldots < f(s_2) = 173116515$.

\end{document}